\documentclass[journal]{IEEEtran}
\usepackage{subcaption, hyperref}

\usepackage{graphicx,amsfonts,amssymb,amsmath,amsthm}
\usepackage{xcolor,cite}
\newcommand{\eps}{\varepsilon}
\newcommand{\E}{{\mathbb{E}}}
\newtheorem{thm}{Theorem} 

\newtheorem{lemma}[thm]{Lemma}

\newtheorem{assumption}{Assumption}

\newtheorem{definition}{Definition}
\theoremstyle{remark}

\newtheorem{example}{Example}

\newtheorem{remark}[thm]{Remark}

\newtheorem{tab}{Table}

\begin{document}
\title{Second-Order Asymptotics of \\   Sequential Hypothesis Testing}
\author{Yonglong~Li
         $\qquad$ Vincent~Y.~F.~Tan, {\em Senior Member, IEEE}
\thanks{Y.~Li is  with the Department of Electrical and Computer Engineering,
National University of Singapore (NUS), Singapore 117583 (e-mail: elelong@nus.edu.sg).  V.~Y.~F.~Tan is  with the Department of Electrical and Computer Engineering and the Department of Mathematics, NUS, Singapore 119076 (e-mail: vtan@nus.edu.sg). } \thanks{The authors are supported by a Singapore National Research Foundation (NRF) Fellowship (R-263-000-D02-281). } \thanks{This paper was presented in part at the 2020 International Symposium on Information Theory (ISIT).}
\thanks{Copyright (c) 2017 IEEE. Personal use of this material is permitted.  However, permission to use this material for any other purposes must be obtained from the IEEE by sending a request to pubs-permissions@ieee.org.}}
\maketitle

\begin{abstract}
We consider the classical sequential binary hypothesis testing problem in which there are two hypotheses governed respectively by distributions $P_0$ and $P_1$ and we would like to decide which hypothesis is true using a sequential test. It is known from the work of Wald and Wolfowitz that as the expectation of the length of the test grows, the optimal type-I and type-II error exponents approach the relative entropies $D(P_1\|P_0)$ and $D(P_0\|P_1)$. We refine this result by considering the optimal backoff---or second-order asymptotics---from the corner point of the achievable exponent region $(D(P_1\|P_0),D(P_0\|P_1))$ under two different constraints on the length of the test (or the sample size). First, we consider a probabilistic constraint in which the probability that the length of test exceeds a prescribed integer $n$ is less than a certain threshold $0<\varepsilon <1$. Second, the expectation of the sample size is bounded by $n$. In both cases, and under mild conditions, the second-order asymptotics  is characterized exactly. Numerical examples are provided to illustrate our results.
\end{abstract}


\section{Introduction}
{\em Hypothesis testing} is one of the most canonical problems in information theory and statistics. In its simplest form, hypothesis testing entails deciding between one of a number of different possible physical phenomena given a fixed number of observations. When further restricted to the binary case in which there are only {\em two} possible phenomena, this is known as {\em binary hypothesis testing}, which though simple has a plethora of applications from communication engineering in which one would like to detect the presence (or identity) of a signal from noisy measurement or radar engineering in which one would like to detect the presence of a potentially  malicious object given noisy observations from imperfect sensors. 

Abstracted into a mathematical problem, binary hypothesis testing involves two hypotheses $H_0$ and $H_1$, respectively characterized by distributions $P_0$ and $P_1$ defined on the same alphabet $\mathcal{X}$. A sample $X$ is drawn from either $P_0$ or $P_1$ and given the sample, the decision maker would like to decide using a (possibly randomized) function  $\delta:\mathcal{X}\to \{0,1\}$ whether it was drawn from $P_0$ (i.e., $H_0$ is in effect) or $P_1$ (i.e., $H_1$ is in effect). Naturally, there are two types of errors, viz.\ the type-I error (resp.\ type-II error) in which the true hypothesis is $H_0$ (resp.\ $H_1$) but the decoder $\delta$ decided based on $X$ that the true hypothesis is $H_1$ (resp.\ $H_0$). The type-I and type-II error probabilities are also known as the {\em false-alarm} and {\em mis-detection} probabilities respectively. 

A regime of interest is when the number of observations is not just one (as above) but can be large and is modelled by mathematically by a number $n$ tending to infinity. In this case, $X$ above is replaced by a vector $X^n = (X_1, \ldots, X_n)$ and a number of pertinent questions arise. What is the optimal tradeoff between the two types of error probabilities? What is the optimal performance in terms of the type-I and type-II error probabilities? The former is dictated by the Neyman-Pearson lemma~\cite{NeymanPearson} which states, roughly speaking, that the likelihood ratio test is optimal. Furthermore, according to the Chernoff-Stein lemma~\cite[Theorem 1.2]{CsiszarKorner}, if we constrain the type-I error to be $\le\varepsilon$ for some $\varepsilon\in (0,1)$, then the best (largest) exponent rate of decay of the type-II error with $n$---also known as the type-II error exponent---is the relative entropy $D(P_0 \| P_1)$. If we instead demand that the type-I error probability decays exponentially fast as $\exp(-n \eta)$ for some $\eta \in (0, D(P_1 \| P_0))$, then the  type-II  error  probability decays as $\min_{Q: D(Q \| P_1) \le \eta} D(Q\| P_0)$~\cite{CsiszarKorner, Tuncel, Blahut}. This can  also be expressed in terms of a R\'enyi divergence~\cite{Blahut}. As can be seen, there is a tradeoff between the two error exponents; both of them cannot be large at the same time. 

Somewhat surprisingly, this pesky tradeoff can be completely eradicated if we allow the length of the test (i.e., the number of observations) to be a stopping time whose {\em expectation} is bounded by $n$. In this case, Wald and Wolfowitz~\cite{WaldWolf} showed that there exists a sequence of tests---namely sequential probability ratio tests (SPRTs)---such that the  type-I and type-II error exponents {\em simulatenously} assume the extremal values  $D(P_1\| P_0)$ and $D(P_0\| P_1)$. 

\subsection{Main Contributions}
In this work, and motivated by practical applications in which the length of the test must be finite, we study the so-called second-order asymptotic regime~\cite{PPV10, Hayashi08} in which the backoff from the corner point of the achievable exponent regime  $(  D(P_1\| P_0),D(P_0\| P_1))$ as a function of the permissible average length of the test $n$ is quantified. This requirement that the stopping time's expectation is not larger than $n$ is coined the {\em expectation constraint} and the backoff we seek is quantified through the {\em second-order exponent under the expectation constraint}.    We also consider a complementary setting in which we constrain the probabilities that the event that the stopping time exceeds $n$ is no larger than some prescribed $\varepsilon\in (0,1)$. The backoff we seek here is quantified through the {\em second-order exponent under the probabilistic constraint}. Under  these constraints, and by assuming mild conditions on the distributions, we derive the backoff, or more precisely, the second-order term, by judiciously combining and utilizing classical tools from probability theory~\cite{CLTmaximal} and nonlinear renewal theory~\cite{durrettprobability, nonlinearrenewaltheory} as well as key properties of the SPRT~\cite{ferguson, SequentialAnalysis ,Lotov86,AsymptoticsSPRT}.  We notice that under the expectation constraint, the convergence to $(  D(P_1\| P_0),D(P_0\| P_1))$ is of the order $\Theta(\frac{1}{n})$, which is much faster than under the probabilistic constraint in which the convergence rate is $\Theta(\frac{1}{\sqrt{n}})$.  

Because  the second-order term under  the expectation constraint is stated rather implicitly, in  Section~\ref{sec:examples} we also suggest a numerical procedure to compute it. This procedure is applied to two examples of discriminating Gaussians and exponential distributions with different means.

\subsection{Related Works}
The literature on sequential hypothesis testing is so vast that we will not attempt to survey all relevant papers here. Rather, we will mention the most pertinent ones, split into those from the information-theoretic and statistics perspectives. 

From the information theory perspective, for the fixed-length case, Strassen~\cite{Strassen} showed using the central limit theorem that the backoff from $D(P_0 \| P_1)$ is of the order $\Theta(\frac{1}{\sqrt{n}})$ and the implied constant was also quantified in terms of the so-called relative entropy variance~\cite{Vincentbook} and the Gaussian cumulative distribution function.  When both error exponents are required to be positive, a strong large deviations analysis implied by the works of Altu\u{g} and Wagner~\cite{Altug2014, Altug20142} can be also used to quantified the backoff, which is of the order $\Theta(\frac{\log n}{n})$. For sequential hypothesis testing, to the best of the authors' knowledge, there has been not much work on backing off from infinity. One related and elegant work that formed the motivation for the current work is that by Lalitha and Javidi~\cite{anusha}. The authors derived the optimal  type-I and type-II error exponents associated with a  newly introduced {\em almost-fixed-length hypothesis test}; this is one in which with exponentially small probability, the decision maker is allowed to collect another set of samples (in addition to the $n$ she already possesses). The authors showed that this setting interpolates between classical hypothesis testing with a fixed sample size and sequential hypothesis testing.  In a related non-Bayesian  setting, Polyanskiy and Verd\'u~\cite{PV10} considered binary hypothesis testing with feedback wherein the two hypothesis are characterized by discrete memoryless channels $P_{Y|X}$ and $Q_{Y|X}$. In addition to being able to access feedback from the tester, the controller is also able to adaptively control the inputs to the controller. It was shown that the   control strategy used in~\cite{PV10} is asymptotically optimal in a certain Bayesian setting studied by Naghshvar and Javidi~\cite{NJ13}.

In the statistics literature~\cite{WaldWolf, ferguson, SequentialAnalysis}, for the sequential hypothesis testing problem, statisticians have extensively studied the performance of the two types of error probabilities and the expected sample length as functions of the parameters of SPRTs. In~\cite{Lotov86, AsymptoticsSPRT}, Lotov studied SPRTs and derived series formulas for the two error probabilities and the expected sample length as the parameters of SPRTs grow without bound assuming mild regularity conditions on the distributions. In contrast, in this work, we study the {\em optimal} second-order exponents of the two types of error probabilities under two families of constraints on the expected sample length over {\em all} sequential hypothesis tests (SHTs) and show that SPRTs asymptotically achieve the second-order exponents. That is, SPRTs are optimal not only in the sense of maximizing the first-order error exponents~\cite{WaldWolf} but also from the perspective of optimizing their  second-order terms, which we characterize in this paper.

\section{Problem Setup}

Let $P_0$ and $P_1$ be two probability measures on $\mathbb{R}$ such that 
$$0<D(P_{0}\|P_{1})<\infty\quad \mbox{and}\quad 
0<D(P_{1}\|P_{0})<\infty,$$  where recall that the {\em relative entropy} $D(P\| Q)= \int\log\frac{\mathrm{d} P}{\mathrm{d} Q}\, \mathrm{d} P$.
Under hypothesis $H_i$ for $i =0,1$, probability measure $P_i$ is in effect.  For $i=0,1$, let $p_i(\cdot)$ be the density (resp.\ probability mass function) of $P_i$  when $P_i$ is absolutely continuous with respect to the Lebesgue measure on $\mathbb{R}$ (resp.\ the support of $P_i$ assumes finitely many values). For $i=0,1$, let $\mu_i$ be the probability measure on $\mathbb{R}^{\infty}$ such that $\mu_i=P_i\times P_i\times\ldots$. Let $X=\{X_i\}_{i=1}^{\infty}$ be an observed i.i.d.\ sequence, where $X_i\sim P_0$  or $X_i\sim P_1$. Let $\mathcal{F}(X_1^n)$ be the $\sigma$-algebra generated by $(X_1,\ldots,X_n)$, the first $n$ random variables of $X$. Let $T$ be a stopping time adapted to the filtration $\{\mathcal{F}(X_1^n)\}_{n=1}^{\infty}$ and let $\mathcal{F}_T$ be the $\sigma$-algebra associated with $T$ (for definitions of stopping times and $\mathcal{F}_T$, see~\cite{durrettprobability}). Let $\delta$ be a $\{0,1\}$-valued $\mathcal{F}_T$-measurable function. The pair $(\delta, T)$ is called an  SHT. In an  SHT $(\delta,T)$, $\delta$ is called the {\em decision function} and $T$ is called the {\em sample size}. When $\delta=0$ (resp.\ $\delta=1$), the decision is made in favor of $H_0$ (resp.\ $H_1$).  Let $P_0^{T}$ (resp.\ $P_1^{T}$) be the probability measure obtained by restricting $\mu_0$ (resp.~$\mu_1$) to $\mathcal{F}_T$. The {\em type-\uppercase\expandafter{\romannumeral1}} and {\em type-\uppercase\expandafter{\romannumeral2} error probabilities} are defined as 
$$
P_{1|0}(\delta,T)=P_0^{T}(\delta=1)\,\mbox{and}\,P_{0|1}(\delta,T)=P_1^{T}(\delta=0).
$$
In other words, $P_{1|0}(\delta,T)$  (resp.\ $P_{0|1}(\delta,T)$) is the error probability that the true hypothesis is $P_0$ (resp.\ $P_1$) but $\delta=1$ (resp.\ $\delta=0$) based on the observations up to time $T$. For notational convenience, when the underlying (expected) length of the sequence $n$ is clear from the context, both $P_{i}^{T}(\cdot)$ and $P_{i}^{n}(\cdot)$ will be abbreviated as $P_{i}(\cdot)$; that is, we omit all superscripts on $P_i (\cdot)$. 

One important class of SHTs is the family of SPRTs. Let $Y_k=\log\frac{p_0(X_k)}{p_1(X_k)}$ and $S_n=\sum_{k=1}^{n}Y_k$. For any pair of positive real numbers $\alpha$ and $\beta$, an SPRT with parameters $(\alpha,\beta)$ is defined as follows
$$
\delta=\begin{cases}0&\mbox{if $
S_T>\beta$}\\
1&\mbox{if $S_T<-\alpha$},
\end{cases}
$$
where $T=\inf\{n\ge 1:S_n\notin[-\alpha,\beta]\}$.

In this paper, we consider two kinds of constraints on the sample size $T$. The first is the {\em probabilistic constraint}; that is, for any integer $n$ and error tolerance $0<\eps<1$, the sample size $T$ satisfies that $\max_{i=0,1}{P_i(T> n)}\le \eps$. The second is the {\em expectation constraint}; that is, for any integer $n$, $\max_{i=0,1}{\E _{P_i}[T]}\le n$. We say that an error exponent pair $(E_0, E_1)$ is {\em achievable under the expectation constraint} if there exists a sequence of  SHTs $\{(\delta_n,T_n)\}_{n=1}^{\infty}$ with $\max_{i=0,1}{\E _{P_i}[T_n]}\le n$ such that 
\begin{align*}
 \liminf_{n\to\infty}\frac{1}{n}\log \frac{1}{P_{1|0}(\delta_n,T_n)}&\ge E_0 ,\quad \mbox{and}\\
  \liminf_{n\to\infty}\frac{1}{n}\log \frac{1}{P_{0|1}(\delta_n,T_n)}&\ge E_1.
\end{align*}
The set of all achievable $(E_0, E_1)$ is denoted as ${\cal E}(P_0,P_1)$. 
It has been shown in~\cite{WaldWolf} that  
$${\cal E}(P_0,P_1) =\{(E_0, E_1):E_0E_1\le
D(P_1\|P_0)D(P_0\|P_1)\}.$$
The error exponent pair $(E_0,E_1)=(D(P_1\|P_0),D(P_0\|P_1))$ can be achieved by a sequence of SPRTs. In this paper, we are concerned with the {\em speed} or {\em rate of convergence} to this   point of the achievable region of the error exponents $(D(P_1\| P_0), D(P_0\| P_1))$ under the two kinds of constraints on the sample size $T_n$. The rates of convergence are formally defined in the following:

\begin{definition} For fixed $(\lambda , \eps)\in[0,1]\times (0,1)$ and an SHT $(\delta_n,T_n)$, 
let \begin{align}
& g_n(\lambda,\eps|\delta_n,T_n)\nonumber\\
 &\;=\lambda\left(\frac{1}{\sqrt{n}}\log \frac{1}{P_{1|0}(\delta_n, T_n)}-\sqrt{n} \,D(P_1\|P_0)\right)\notag\\
&\quad+(1\!-\!\lambda)\left(\frac{1}{\sqrt{n}}\log \frac{1}{P_{0|1}(\delta_n, T_n)}\! -\!\sqrt{n}\,D(P_0\|P_1)\right).\notag
\end{align}
and 
\begin{align}\label{probabilistic1}
\hspace{-6mm}G_n(\lambda,\eps)&=\sup_{ (\delta_n, T_n):\max_{i=0,1}P_i( T_n> n)\le\eps} g_n(\lambda,\eps|\delta_n,T_n).
\end{align}
Let $
\overline{G}(\lambda,\eps)=\limsup_{n\to\infty}G_n(\lambda,\eps)$ and $\underline{G}(\lambda,\eps)=\liminf_{n\to\infty}G_n(\lambda,\eps).
$
If 
$\overline{G}(\lambda,\eps)=\underline{G}(\lambda,\eps)$, then we term this common value as the {\em second-order exponent of SHT under the probabilistic constraint} and we denote it simply as $G(\lambda,\varepsilon)$.  
\end{definition}

\begin{definition} For fixed $\lambda\in[0,1]$ and an SHT $(\delta_n,T_n)$, 
let
\begin{align}
f_n(\lambda|\delta_n,T_n)&=\lambda\left(\log {P_{1|0}(\delta_n, T_n)}+nD(P_1\|P_0)\right)\notag\\
&\quad+(1-\lambda)\left(\log {P_{0|1}(\delta_n, T_n)}+nD(P_0\|P_1)\right)\notag
\end{align}
and
 \begin{align}\label{expectationconstraint}
F_n(\lambda)&=\sup_{ (\delta_n, T_n):\max_{i=0,1}\E_{P_i}[ T_n]\le  n} f_n(\lambda|\delta_n,T_n).
\end{align}
Let $\overline{F}(\lambda)=\limsup_{n\to\infty}F_n(\lambda)$ and $\underline{F}(\lambda)=\liminf_{n\to\infty}F_n(\lambda).
$
If $\overline{F}(\lambda)=\underline{F}(\lambda)$, then we term this common value as the {\em second-order exponent of SHT under the expectation constraint} and we denote it simply as $F(\lambda)$.\footnote{We note that $G_n(\lambda, \varepsilon)$ and $F_n(\lambda)$ in \eqref{probabilistic1} and \eqref{expectationconstraint} respectively are defined with opposing signs; this is to ensure that the results are stated as cleanly as possible.  The normalizations of $g_n(\lambda,\eps|\delta_n,T_n)$ and $f_n(\lambda|\delta_n,T_n)$ by different functions of the expected length $n$  are also different. }  
\end{definition}

Throughout the paper, $\Phi(\cdot)$ is used to denote the cumulative distribution  function of a standard Gaussian  and $\Phi^{-1}(\cdot)$ is used to denote the generalized inverse   of $\Phi(\cdot)$.
\section{Main Results}
Our first theorem characterizes the second-order exponent under the probabilistic constraint on the sample size. For $i=0,1$, we define the {\em relative entropy variance}~\cite{Vincentbook} between $P_i$ and $P_{1-i}$ as 
\begin{align}
V(P_i\|P_{1-i})&=\E_{P_i}\left[\bigg(\log\frac{p_{i}(X_1)}{p_{1-i}(X_1)}-D(P_i\|P_{1-i})\bigg)^2\right]\notag\\
&= \mathrm{Var}_{ P_i}\left[\log\frac{p_{i}(X_1)}{p_{1-i}(X_1)} \right].\notag
\end{align}
\begin{thm}\label{probabilistic-thm}
Let $P_0$ and $P_1$ be such that 
\begin{align}\label{thirdmoment}
\max_{i=0,1}\E_{P_i}\left[\Big|\log\frac{p_0(X_1)}{p_{1}(X_1)}\Big|^3\right]<\infty.
\end{align} Then,  for every  $\lambda\in[0,1]$ and $0<\eps<1$, we have
\begin{align}\label{probabilistic}
 G(\lambda,\eps) &=\overline{G}(\lambda,\eps)=\underline{G}(\lambda,\eps)\notag\\
&= \lambda \sqrt{V(P_0\|P_1)}\Phi^{-1}(\eps) \notag\\
&\qquad+(1-\lambda)\sqrt{V(P_1\|P_0)}\Phi^{-1}(\eps).
\end{align}
\end{thm}
The proof of Theorem~\ref{probabilistic-thm} can be found in Section~\ref{sec:prf_prob} and essentially relies  on the Berry-Esseen theorem for the maximal sum of independent random variables~\cite{CLTmaximal}.

Our second main theorem concerns the second-order exponent under the expectation constraint.  To set things up before stating our result, we introduce a few definitions. A real-valued random variable $Z$ is said to be {\em arithmetic} if there exists a positive real number $d$ such that $P(Z\in d\mathbb{Z})=1$; otherwise $Z$ is said to be {\em non-arithmetic}. If $Z$ is arithmetic, then the smallest positive $d$ such that $P(Z\in d\mathbb{Z})=1$ is called the {\em span} of $Z$. 


Let $\{\alpha_k\}_{k=1}^{\infty}$ and $\{\beta_k\}_{k=1}^{\infty}$ be two increasing sequences of positive real numbers such that $\alpha_k\to \infty$ and $\beta_k\to \infty$ as $k\to \infty$. Let ${T}(\beta_k)=\inf\{n\ge 1:S_n>\beta_k\}$ and $\tilde{T}(\alpha_k)=\inf\{n\ge 1:-S_n>\alpha_k\}$. Furthermore, let $R_k=S_{T(\beta_k)}-\beta_k$ and $\tilde{R}_k=-S_{\tilde{T}(\alpha_k)}-\alpha_k$. From~\cite[Theorem~2.3, pp.~18]{nonlinearrenewaltheory}, it follows that 
\begin{itemize}
\item
if the true hypothesis is $H_0$, $\{ {R}_k\}_{k=1}^{\infty}$ converges in distribution to some random variable $ {R}$ and the limit is independent of the choice of $\{\alpha_k\}_{k=1}^{\infty}$;
\item 
if the true hypothesis is $H_1$, $\{\tilde{R}_k\}_{k=1}^{\infty}$ converges in distribution to some random variable $\tilde{R}$ and the limit is independent of the choice of $\{\beta_k\}_{k=1}^{\infty}$.
\end{itemize}
Define 
\begin{alignat}{2}
A(P_0,P_1) &=\E[ {R}],&\qquad\tilde{A}(P_0,P_1) &=\E[\tilde{R}], \notag\\
  B(P_0,P_1)&=\log\E[e^{-R}],&\qquad\tilde{B}(P_0,P_1)&=\log\E[e^{-\tilde{R}}].\notag
\end{alignat}
We note that  these quantities are, in general, not symmetric in their arguments, i.e., $\tilde{A}(P_0,P_1 )\ne A(P_1,P_0)$ and $\tilde{B}(P_0,P_1)\ne  B(P_1,P_0)$ in general.
 \begin{thm}\label{expectation-thm} 
Let $P_0$ and $P_1$ be such that $$\max_{i=0,1}\E_{P_i}\left[\Big|\log\frac{p_0(X_1)}{p_1(X_1)}\Big|^2\right]<\infty$$ and $\log\frac{p_0(X_1)}{p_1(X_1)}$ is non-arithmetic when $X_1\sim P_0$. Then   for every  $\lambda\in[0,1]$, 
 \begin{align}\label{thm2}
F(\lambda)&= \overline{F}(\lambda)=\underline{F}(\lambda)\notag\\
&=\lambda \big(\tilde{A}(P_0,P_1)+\tilde{B}(P_0,P_1)\big)\notag\\
&\hspace{0.4cm}+(1-\lambda)\big(A(P_0,P_1)+ B(P_0,P_1)\big).
\end{align}
\end{thm}
Theorem \ref{expectation-thm} is proved in Section~\ref{sec:prf_exp} and  relies on   results in nonlinear renewal theory~\cite{durrettprobability, nonlinearrenewaltheory} as well as key properties of the SPRT~\cite{ferguson, SequentialAnalysis ,Lotov86,AsymptoticsSPRT}.
 
\begin{remark}
If $X_1\sim P_0$ and $\log\frac{p_0(X_1)}{p_1(X_1)}$  is arithmetic, say with span $d>0$, Theorems~\ref{asymptotic-nonlattice} and~\ref{asymptoticserrornonlattice} (see Section~\ref{tools}) that we leverage on in the proof of Theorem \ref{expectation-thm} hold only for SPRTs with parameters $(\alpha_n,\beta_n)$ where $\alpha_n$ and $\beta_n$ are integer multiples of $d$. In this case $(\alpha_n,\beta_n)$ should be chosen as $(\lfloor{\frac{nD(P_1\|P_0)}{d}-a_0\rfloor}d,\lfloor{\frac{nD(P_0\|P_1)}{d}-a_1\rfloor}d)$ for some constants $a_0$ and $a_1$ depending on $P_0$ and $P_1$. However, the limit of $(\lfloor{\frac{nD(P_1\|P_0)}{d}-a_0\rfloor}d,\lfloor{\frac{nD(P_0\|P_1)}{d}-a_1\rfloor}d)$ as $n\to\infty$ may not exist. This technical difficulty makes the problem challenging for arithmetic $\log\frac{p_0(X_1)}{p_1(X_1)}$ and we defer the analysis of this case to future work. 
\end{remark}

\begin{remark}
From Theorems \ref{probabilistic-thm} and \ref{expectation-thm}, we see that the rate of convergence of the optimal $\lambda$-weighted finite-length exponents $\sup_{(\delta_n,T_n)}-\frac{\lambda}{n} \log {P_{1|0} (\delta_n,T_n)} -\frac{1-\lambda}{n} \log {P_{0|1} (\delta_n,T_n)} $ to the $\lambda$-weighted exponents $\lambda D(P_1\| P_0)+ (1-\lambda ) D(P_0\|P_1)$ is faster under the
expectation constraint as compared to the probabilistic constraint. In fact, the former rate is $\Theta(\frac{1}{n})$ while the latter rate is $\Theta(\frac{1}{\sqrt{n}})$. Thus, the latter constraint is more stringent. Our main contribution is to nail down the exact order and the constants of $\Theta(\frac{1}{\sqrt{n}})$ and $\Theta(\frac{1}{n})$ given in~\eqref{probabilistic} and~\eqref{thm2} respectively.  The expectation constraint  is   reminiscent of variable-length channel coding with feedback~\cite{PPV11, TT17} in which the speed of convergence to the $\varepsilon$-capacity is not $\Theta(\frac{1}{\sqrt{n}})$ but rather $O(\frac{\log n}{n})$. However, different from~\cite{PPV11, TT17}, the ``strong converse'' holds as the first-order fundamental limit $(D(P_1\| P_0), D(P_0\| P_1))$ does not depend on $\varepsilon$.
\end{remark}
\subsection{Examples} \label{sec:examples}
In this subsection, we present two examples and numerically compute their second-order exponents under the expectation constraint (since the computation of second-order terms under the probabilistic constraint is straightforward). The following theorem extracted from~\cite[Corollary~2.7 and Theorem~3.3, pp.~24 and~32]{nonlinearrenewaltheory} provides  more concrete characterizations of  the quantities $A(P_0,P_1), \tilde{A}(P_0,P_1), B(P_0,P_1)$, and $\tilde{B}(P_0,P_1)$. For a real number $a$, we write $a^+=\max\{a,0\}$ and $a^{-}=-\min\{a, 0\}$.

\begin{thm}\label{computation}
Let $P_0$ and $P_1$ be such that $$\max_{i=0,1}\E_{P_i}\left[\Big|\log\frac{p_0(X_1)}{p_1(X_1)}\Big|^2\right]<\infty$$ and $\log\frac{p_0(X_1)}{p_1(X_1)}$ is non-arithmetic when $X_1\sim P_0$. Then
 \begin{align}
 A(P_0,P_1)&=\frac{\E_{P_0}\left[\log^2\frac{p_0(X_1)}{p_1(X_1)}\right]}{2D(P_0\|P_1)}-\sum_{k=1}^{\infty}\frac{1}{k}\E_{P_0}\left[S_k^{-}\right], \notag\\
 \tilde{A}(P_0,P_1)&=\frac{\E_{P_1}\left[\log^2\frac{p_0(X_1)}{p_1(X_1)}\right]}{2D(P_1\|P_0)}-\sum_{k=1}^{\infty}\frac{1}{k}\E_{P_1}\left[S_k^{+}\right],\notag
\end{align}
and 
\begin{align}
 B(P_0,P_1)&=-\log D(P_0\|P_1)\notag\\
 &\quad-\sum_{k=1}^{\infty}\frac{1}{k}\big(P_{0}(S_k<0)+P_1(S_k>0)\big),\notag\\
 \tilde{B}(P_0,P_1)&=-\log D(P_1\|P_0)\notag\\
 &\quad-\sum_{k=1}^{\infty}\frac{1}{k}\big(P_{0}(S_k<0)+P_1(S_k>0)\big).\notag
\end{align}
\end{thm}

\begin{example}[Two Gaussians] \label{ex:gauss} Let $\theta_0$ and $\theta_1$ be two distinct real numbers. Let $p_0(x)=\frac{1}{\sqrt{2\pi}}e^{-\frac{(x-\theta_0)^2}{2}}$  and $p_1(x)=\frac{1}{\sqrt{2\pi}}e^{-\frac{(x-\theta_1)^2}{2}}$ for $x\in\mathbb{R}$. Then $Y_k=\log\frac{p_0(X_k)}{p_1(X_k)}=\frac{\theta_1^2-\theta_0^2}{2}+(\theta_0-\theta_1) X_k$. Let $\Delta\theta=\theta_1-\theta_0$ be the difference of the means. Thus, under   hypothesis $H_0$, $Y_k\sim \mathcal{N} \big(\frac{(\Delta\theta)^2}{2},(\Delta\theta)^2 \big)$, and under $H_1$, $Y_k\sim \mathcal{N}\big(-\frac{(\Delta\theta)^2}{2}, (\Delta\theta)^2\big)$. Therefore, under   hypothesis $H_0$, $S_k\sim \mathcal{N}\big(\frac{k(\Delta\theta)^2}{2},k(\Delta\theta)^2\big)$, and under    $H_1$, $S_k\sim \mathcal{N}\big(-\frac{k(\Delta\theta)^2}{2},k(\Delta\theta)^2\big)$. We can then derive the following important quantities: 
\begin{itemize}
\item $D(P_0\|P_1)=D(P_1\|P_0)=\frac{(\Delta\theta)^2}{2}$;
\item $\begin{aligned}[t]\E_{P_0}\left[\log^2\frac{p_0(X_1)}{p_1(X_1)}\right]&=\E_{P_1}\left[\log^2\frac{p_0(X_1)}{p_1(X_1)}\right]\\
&=\frac{(\Delta\theta)^4}{4}+(\Delta\theta)^2\mbox{;}\end{aligned}$
\item $P_{0}(S_k<0)+P_1(S_k>0)=2\Phi\left(-\frac{\sqrt{k}|\Delta\theta|}{2}\right);$
\item $\begin{aligned}[t]&\E_{P_1}\left[S_k^{+}\right]=\E_{P_0}\left[S_k^{-}\right]=-\frac{k(\Delta\theta)^2}{2}\Phi\bigg(-\frac{\sqrt{k}|\Delta\theta|}{2}\bigg)\\
&\hspace{4cm}+\frac{\sqrt{k|\Delta\theta|}}{2\pi}e^{-\frac{k|\Delta\theta|}{8}}.
\end{aligned}$  
\end{itemize}
Using Theorems~\ref{expectation-thm} and~\ref{computation}, we can numerically compute the second-order exponent under the expectation constraint in \eqref{thm2}. This is illustrated in Figure~\ref{fig:1}. We note that for this  case of discriminating between two Gaussians, $F(\lambda)$ does not depend on $\lambda\in [0,1]$. 
\end{example} 

\begin{figure}
\centering
  \begin{subfigure}{0.45\textwidth}
  \includegraphics[width=7.0cm]{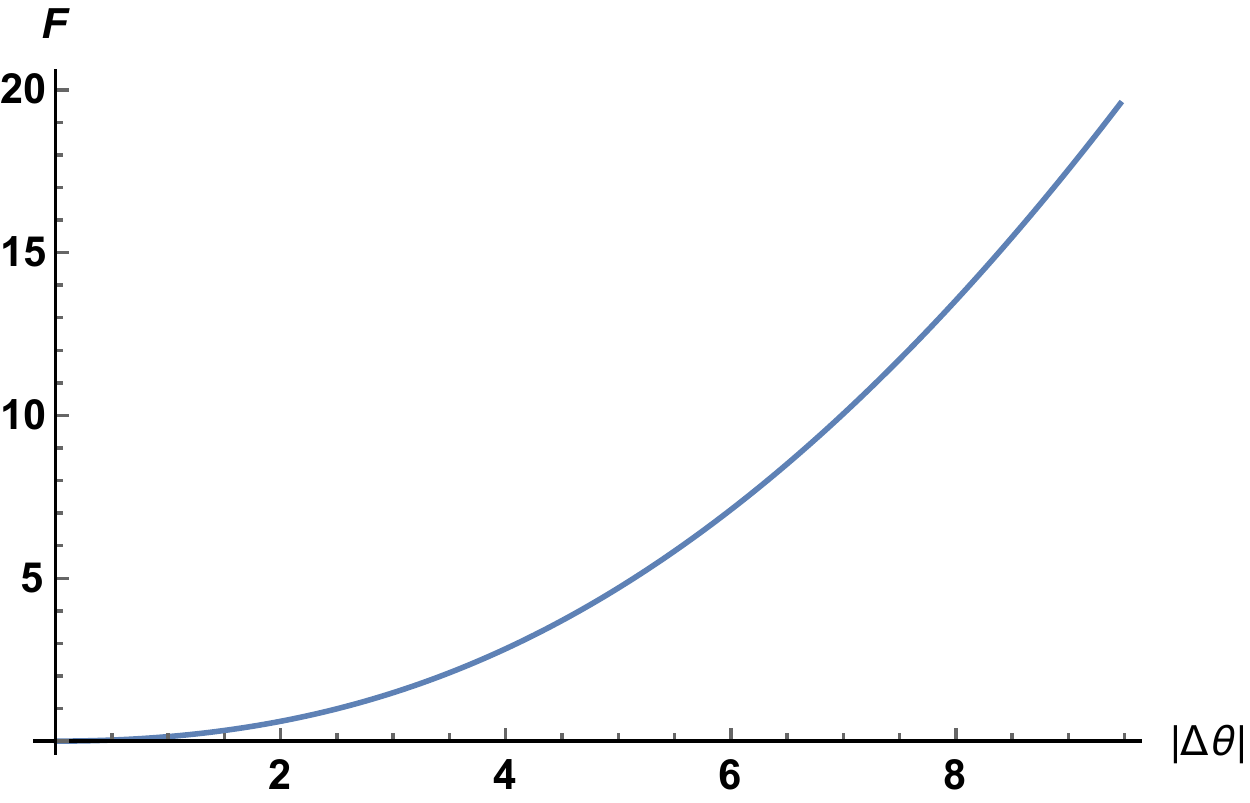}
    \caption{Illustration of $F(\lambda)$ for two  Gaussian distributions as in Example~\ref{ex:gauss}}
    \label{fig:1}
  \end{subfigure}
\hspace{1cm}
  \begin{subfigure}{0.45\textwidth}
    \includegraphics[width=7.0cm]{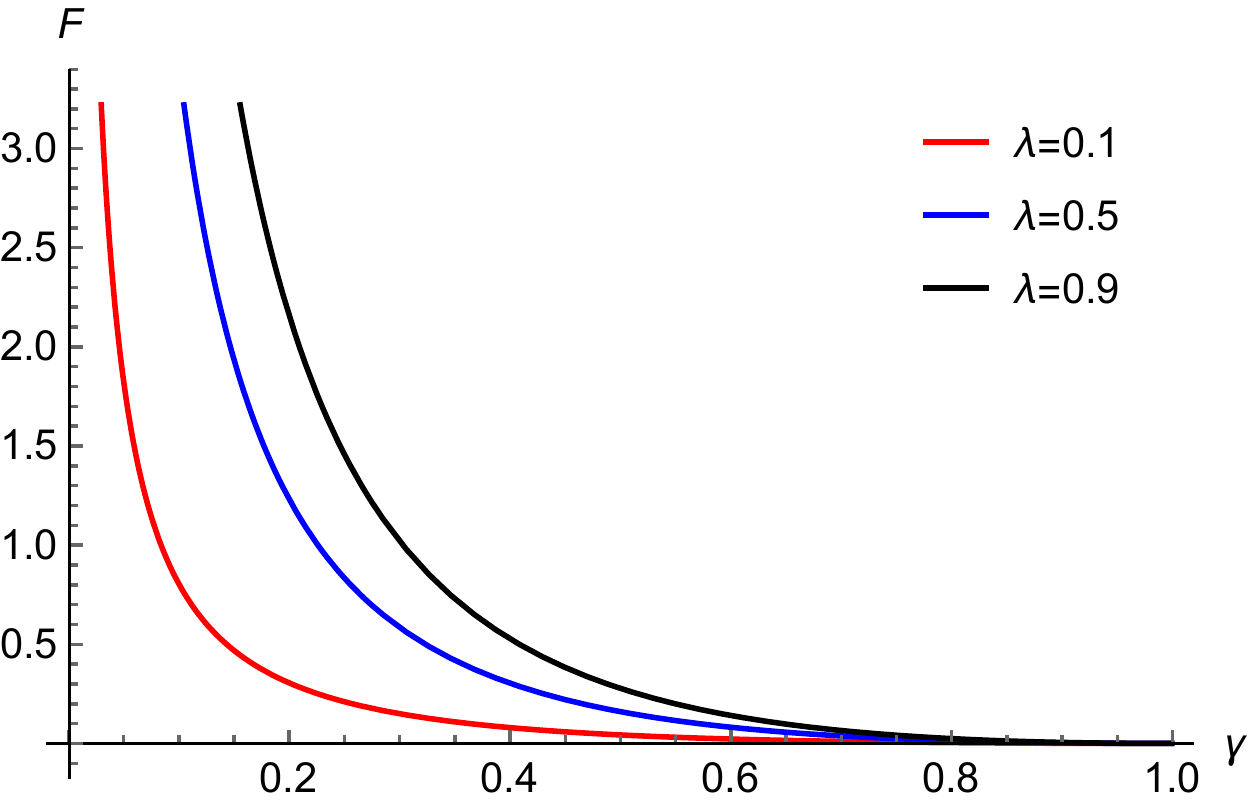}
    \caption{Illustration of $F(\lambda)$ for two exponential distributions as in Example~\ref{ex:exp} with $\gamma_0=\gamma$ and $\gamma_1=1$}
    \label{fig:2}
  \end{subfigure}
  \caption{Illustration of the second-order exponents under the expectation constraint }
  \end{figure}

\begin{example}[Two Exponentials] \label{ex:exp} Let $\gamma_0$ and $\gamma_1$ be two positive real numbers such that $\gamma_0<\gamma_1$. Let $p_0(x)=\gamma_0 e^{-\gamma_0 x}$ and $p_1(x)=\gamma_1 e^{-\gamma_1 x}$ for $x>0$. Then $Y_k=(\gamma_1-\gamma_0)X_k+\log \frac{\gamma_0}{\gamma_1}$ and $S_k=(\gamma_1-\gamma_0)\sum_{i=1}^{k}X_i+k\log \frac{\gamma_0}{\gamma_1}$. Let $x \in [0,\infty)\mapsto U(x;k,\gamma)$ denote the cumulative distribution function of the Erlang distribution with parameters $(k,\gamma)$.

In~\cite[Section~3.1.6]{SequentialAnalysis}, by using   tools which we will not discuss here,  exact formulas for $A(P_0,P_1)$, $\tilde{A}(P_0,P_1), B(P_0,P_1)$ and $ \tilde{B}(P_0,P_1) $ for this example concerning exponential distributions were derived. Here, we use  Theorem~\ref{computation}  together with the fact that the sum of $k$ i.i.d.\ exponential random variables with parameter $\gamma$ is the Erlang distribution with parameters $(k,\gamma)$  to derive a formula  for the second-order term in~\eqref{thm2}.
We can derive the following important quantities: 
\begin{itemize}
\item $D(P_0\|P_1)=\log\frac{\gamma_0}{\gamma_1}+\frac{\gamma_1-\gamma_0}{\gamma_0}$ and $D(P_1\|P_0)=\log\frac{\gamma_1}{\gamma_0}+\frac{\gamma_0-\gamma_1}{\gamma_1}$;
\item $\E_{P_0}\left[\log^2\frac{p_0(X_1)}{p_1(X_1)}\right]=\Big(\log\frac{\gamma_0}{\gamma_1}+\frac{\gamma_1-\gamma_0}{\gamma_0} \Big)^2+\frac{(\gamma_0-\gamma_1)^2}{\gamma_0^2}$ and  \\ $\E_{P_1}\left[\log^2\frac{p_0(X_1)}{p_1(X_1)}\right]=\Big(\log\frac{\gamma_1}{\gamma_0}+\frac{\gamma_0-\gamma_1}{\gamma_1} \Big)^2+\frac{(\gamma_1-\gamma_0)^2}{\gamma_1^2}$;
\item $P_{0}(S_k <0) + P_1(S_k>0)\!=\!1+U\left(\frac{k}{\gamma_1-\gamma_0}\log\frac{\gamma_1}{\gamma_0};k,\gamma_0\right)-U\left(\frac{k}{\gamma_1-\gamma_0}\log\frac{\gamma_1}{\gamma_0};k,\gamma_1\right);$ 
\item$\E_{P_0}\left[S_k^{-}\right]=k U\left(\frac{k}{\gamma_1-\gamma_0}\log\frac{\gamma_1}{\gamma_0};k,\gamma_0\right)\log\frac{\gamma_1}{\gamma_0}\\-\frac{k(\gamma_1-\gamma_0)}{\gamma_0}U\left(\frac{k}{\gamma_1-\gamma_0}\log\frac{\gamma_1}{\gamma_0};k+1,\gamma_0\right)$ and\vspace{.05in} \\   
$\E_{P_1}\left[S_k^{+}\right] = k\left(1 \! -\! U\left(\frac{k}{\gamma_1-\gamma_0}\log\frac{\gamma_1}{\gamma_0};k,\gamma_1\right)\right)\log\frac{\gamma_0}{\gamma_1}\\+\frac{k(\gamma_1-\gamma_0)}{\gamma_1}\left(1\! -\!  U\left(\frac{k}{\gamma_1-\gamma_0}\log\frac{\gamma_1}{\gamma_0};k + 1,\gamma_1\right)\right)$.
\end{itemize}
Using these facts, together with Theorems~\ref{expectation-thm} and~\ref{computation}, we can numerically compute the second-order exponent under the expectation constraint. This is illustrated in Figure~\ref{fig:2} for various $\lambda$'s. 
\end{example} 


\section{Proof of Theorem~\ref{probabilistic-thm}} \label{sec:prf_prob}
In this section we prove Theorem~\ref{probabilistic-thm} by establishing upper and lower bounds on error probabilities, respectively. Throughout this section we use $\chi_E$ to denote the indicator function taking the value $1$ on the set $E$ (and $0$ otherwise). In the proof, we use the {\em Berry-Esseen theorem for the maximal sum of independent random variables}~\cite{CLTmaximal}. For ease of reference, we restate it here as follows.
\begin{thm}[Rogozin~\cite{CLTmaximal}]\label{CLTmaximal}
Let $X=\{X_i\}_{i=1}^{\infty}$ be an i.i.d.\ sequence of random variables such that $\E[X_1]>0$ and $\E[|X_1|^3]<\infty$. Let $S_k=\sum_{i=1}^{k}X_i$. Then for all $n\ge 1$,
\begin{align*}
 \sup_{a\in\mathbb{R}}\left|{P\bigg(\frac{(\max_{1\le k\le n}S_k)-n\E[X_1]}{\sqrt{n\mathrm{Var}(X_1)}}\le a\bigg)-\Phi(a)}\right|\le \frac{M}{\sqrt{n}},
\end{align*}
where $M$ is a constant depending only on $\mathrm{Var}(X_1)$ and $\E[|X_1|^3]$.
\end{thm}
\subsection{Upper Bound on the Error Probabilities}
For any $\eta \in(0,\eps)$, let $c_0=-\sqrt{V(P_1\|P_0)}\Phi^{-1}(\eps-\eta)$ and $c_1=-\sqrt{V(P_0\|P_1)}\Phi^{-1}(\eps-\eta)$. Let $\alpha_n=n(D(P_1\|P_0)-\frac{c_0}{\sqrt{n}})$ and $\beta_n=n(D(P_0\|P_1)-\frac{c_1}{\sqrt{n}})$ and let $(\delta_n,T_n)$ be the SPRT with parameters $(\alpha_n,\beta_n)$. As $0<D(P_i\|P_{1-i})<\infty$ for $i=0,1$, we know from~\cite[Lemma~3.1, pp.~29]{nonlinearrenewaltheory} that $P_i(T_n<\infty)=1$ for each $n\in\mathbb{N}$. Then using~\cite[Theorem~1.1, pp.~4] {nonlinearrenewaltheory}, we have that
\begin{align}\label{upper}
P_{1|0}(\delta_n,T_n)&=P_0(S_{T_n}<-\alpha_n)\notag\\
&=\E_{P_1}\big[\chi_{\{S_{T_n}<-\alpha_n\}}e^{S_{T_n}} \big]\le e^{-\alpha_n}.
\end{align}
Similarly, we have $P_{0|1}(\delta_n,T_n)\le e^{-\beta_n}.$ Note that
\begin{align}
P_0(T_n>n)&=P_0\left(\alpha_n<\min_{1\le i\le n}S_i\le \max_{1\le i\le n}S_i\le \beta_n\right)\notag\\
&\le P_0\left( \max_{1\le i\le n}S_i\le \beta_n\right)\notag\\
&\le\Phi\bigg(\frac{-c_1}{\sqrt{V(P_0\|P_1)}}\bigg)+\frac{M}{\sqrt{n}}\label{im5}\\
&=\eps-\eta+\frac{M}{\sqrt{n}},\notag
\end{align}
where~(\ref{im5}) holds due to Theorem~\ref{CLTmaximal}  ($X_i$ in Theorem~\ref{CLTmaximal} is taken to be the log-likelihood ratio $\log\frac{p_0(X_i)}{p_1(X_i)}$) noting that $D(P_0\|P_1)>0$  and~$M$ is a positive finite constant due to the finiteness of the third absolute moments of the log-likelihood ratios as assumed in~(\ref{thirdmoment}).

Therefore for sufficiently large $n$, $P_0(T_n>n)\le \eps-\eta/2<\eps$.
Similarly, $P_1(T_n>n)< \eps$ for sufficiently large $n$. Hence we obtain that
\begin{align*}
\underline{G}(\lambda,\eps)&=\liminf_{n\to\infty}  G_n(\lambda,\eps)\\
&\ge \lambda \sqrt{V(P_1\|P_0)}\Phi^{-1}(\eps-\eta) \\
&\qquad+(1-\lambda)\sqrt{V(P_0\|P_1)}\Phi^{-1}(\eps-\eta),
\end{align*}
which further implies  by  
letting $\eta\to 0^{+}$ that
\begin{align*}
\underline{G}(\lambda,\eps)&\ge \lambda \sqrt{V(P_0\|P_1)}\Phi^{-1}(\eps)\\
&\hspace{1.5cm}+(1-\lambda)\sqrt{V(P_1\|P_0)}\Phi^{-1}(\eps),
\end{align*}
as desired.
\subsection{Lower Bound on the Error Probabilities}
In this section, we derive the lower bound on the error probabilities. First,  we slightly generalize~\cite[Lemma~9.2]{WuPolyanskiy}.

\begin{lemma}\label{probabilisticerror} 
Let $(\delta,T)$ be an SHT such that $\min_{i=0,1}P_i(T<\infty)=1$. Then for any set $E\in \mathcal{F}_T$ and $\gamma>0$, we have $P_0(E)-\gamma P_1(E)\le P_{0}(S_T\ge \log \gamma)$.
\end{lemma}
\begin{proof}
As $P_1(T<\infty)=1$, then from~\cite[Theorem~1.1, pp.~4]{nonlinearrenewaltheory} it follows that $P_1^{T}$ is absolutely continuous with respect to $P_0^{T}$ and the Radon-Nikodym derivative $\frac{\mathrm{d}P_1^{T}}{\mathrm{d}P_0^{T}}=e^{-S_T}$.
Therefore we have that for any $E\in\mathcal{F}_T$, $P_1(E)=\E_{P_0}[\chi_{E}e^{-S_T}]$.  
 This further implies that
\begin{align}
P_0(E)-\gamma P_1(E)&=\E_{P_0}[\chi_{E}]-\gamma\E_{P_1}[\chi_{E}]\notag\\
&=\E_{P_0}[\chi_{E}]-\gamma \E_{P_0}\left[\chi_{E}e^{-S_T}\right]\notag\\
&= \E_{P_0}\left[\chi_{E}(1-\gamma e^{-S_T}) \right]\notag\\
&\le  \E_{P_0} \left[\chi_{E\cap \{S_T\ge \log\gamma\}}(1-\gamma e^{-S_T}) \right]\notag\\
&\le \E_{P_0}\left[\chi_{E\cap \{S_T\ge \log\gamma\}} \right]\notag\\
&\le P_{0}(S_T\ge \log \gamma),\notag
\end{align}
as desired.
\end{proof}
\begin{remark}
In~\cite[Lemma 9.2]{WuPolyanskiy}, the authors proved a similar result for the {\em fixed-length} hypothesis testing problem.
\end{remark}
Let $\{(\delta_n, T_n)\}_{n=1}^{\infty}$ be a sequence of SHTs such that $\max_{i=0,1}P_i({T_n>n})\le \eps$ and let $A_i (T_n)=\{\delta_n=i\}$ for $i=0,1$. Then $P_{1|0}(\delta_n, T_n)=P_{0}(A_1(T_n))$ and $P_{0|1}(\delta_n, T_n)=P_1(A_0 (T_n))$. Using Lemma~\ref{probabilisticerror} with $E=A_0 (T_n)$, we have that 
 \begin{align}
&1-P_{1|0}(\delta_n, T_n)-\gamma P_{0|1}(\delta_n, T_n)\notag\\
&\hspace{0.5cm}\le P_{0}(S_{T_n}\ge \log \gamma)\notag\\
&\hspace{0.5cm}\le P_{0}(S_{T_n}\ge \log \gamma, T_n\le n)+P_0(T_n>n)\notag,
\end{align}
which further implies that
\begin{align}\label{im4}
P_0(T_n>n)&\ge 1-P_{1|0}(\delta_n, T_n)-\gamma P_{0|1}(\delta_n, T_n)\notag\\
&\hspace{0.65cm}- P_{0}(S_{T_n}\ge \log \gamma, T_n\le n)\notag\\
&\ge 1-P_{1|0}(\delta_n, T_n)-\gamma P_{0|1}(\delta_n, T_n)\notag\\
&\hspace{0.65cm}- P_{0}\Big(\max_{1\le i\le n}S_i\ge \log \gamma \Big)\notag\\
&\ge 1-P_{1|0}(\delta_n, T_n)-\gamma P_{0|1}(\delta_n, T_n)-\frac{M}{\sqrt{n}}\notag\\
&\hspace{0.65cm}-\left(1-\Phi  \bigg(\frac{\log \gamma-n D(P_0\|P_1)}{\sqrt{n V(P_0\|P_1)}}\bigg)\right)\\
&\ge -P_{1|0}(\delta_n, T_n)-\gamma P_{0|1}(\delta_n, T_n)-\frac{M}{\sqrt{n}}\notag\\
&\hspace{0.5cm}+\Phi \bigg(\frac{\log \gamma-n D(P_0\|P_1)}{\sqrt{n V(P_0\|P_1)}}\bigg),\notag
\end{align}
where~(\ref{im4}) follows from Theorem~\ref{CLTmaximal} and $M$ is the constant in Theorem~\ref{CLTmaximal}. Again, $M$ is finite because we assume~(\ref{thirdmoment}) holds.

Let $\log\gamma= n\big(D(P_0\|P_1)+\Phi^{-1}(\eps+\frac{2M}{\sqrt{n}})\sqrt{\frac{V(P_0\|P_1) }{n}} \big)  $. As $P_{0}(T_n>n)\le \eps$, we then have that
\begin{align*}
\eps&\ge P_{0}(T_n>n)\\
&\ge -P_{1|0}(\delta_n, T_n)-\gamma P_{0|1}(\delta_n, T_n)\\
&\hspace{0.5cm}+\Phi \bigg(\frac{\log \gamma-n D(P_0\|P_1)}{\sqrt{n V(P_0\|P_1)}}\bigg)-\frac{M}{\sqrt{n}}\\
&\ge -P_{1|0}(\delta_n, T_n)-\gamma P_{0|1}(\delta_n, T_n)+\eps+\frac{M}{\sqrt{n}},
\end{align*}
which implies that
\begin{align*}
&\frac{1}{n}\log P_{0|1}(\delta_n, T_n)\ge -\frac{\log\gamma}{n}+\frac{\log \left(\frac{M}{\sqrt{n}}-P_{1|0}(\delta_n, T_n)\right)}{n}\notag\\
&=-D(P_0\|P_1)-\sqrt{\frac{V(P_0\|P_1)}{n}}\Phi^{-1}\left(\eps+\frac{2M}{\sqrt{n}}\right)\\
&\hspace{0.5cm}+\frac{\log \left(\frac{M}{\sqrt{n}}-P_{1|0}(\delta_n, T_n)\right)}{n}\notag\\
&=-D(P_0\|P_1)-\sqrt{\frac{V(P_0\|P_1)}{n}}\Phi^{-1}(\eps)\\
&\hspace{0.5cm}+\frac{\log \left(\frac{M}{\sqrt{n}}-P_{1|0}(\delta_n, T_n)\right)}{n}+O\bigg(\frac{1}{n}\bigg).
\end{align*}
Let $\log t_n=-n(\min\{D(P_0\|P_1),D(P_1\|P_0)\}/2)$ and 
$$
\mathcal{T}^{(n)}=\left\{(\delta_n,T_n): \begin{array}{c}\max_{i=0,1}P_i({T_n>n})\le \eps\\ P_{1|0}(\delta_n, T_n)\le t_n\\
P_{0|1}(\delta_n, T_n)\le t_n\end{array}\right\}.
$$ For any $(\delta_n,T_n)\in \mathcal{T}^{(n)}$, we know $P_{1|0}(\delta_n,T_n)$ tends to zero exponentially fast, hence we have that 
$$\lim_{n\to\infty}\frac{\log \left(\frac{M}{\sqrt{n}}-P_{1|0}(\delta_n, T_n)\right)}{\sqrt{n}}= 0.$$ Thus it then follows that for any $(\delta_n,T_n)\in \mathcal{T}^{(n)}$,
\begin{align}\label{im8}
&\limsup_{n\to\infty}\sqrt{n}\left(-\frac{1}{n}\log P_{0|1}(\delta_n, T_n)-D(P_0\|P_1)\right)\notag\\
&\hspace{3.6cm}\le \sqrt{V(P_0\|P_1)}\Phi^{-1}(\eps).
\end{align}
Similarly, we have that for any $(\delta_n,T_n)\in \mathcal{T}^{(n)}$,
\begin{align}\label{im9}
&\limsup_{n\to\infty}\sqrt{n}\left(-\frac{1}{n}\log P_{1|0}(\delta_n, T_n)-D(P_1\|P_0)\right)\notag\\
&\hspace{3.6cm}\le \sqrt{V(P_1\|P_0)}\Phi^{-1}(\eps).
\end{align}
From~(\ref{upper}) we know that for sufficiently large $n$, the sequence of exponents of the SHTs that approaches the supremum in~(\ref{probabilistic1}) is lower bounded by $\min\{D(P_0\|P_1),D(P_1\|P_0)\}/2$.  By the (strict) positivity of the relative entropies, it follows that for sufficiently large $n$, the sequence of SHTs that approaches the supremum in~(\ref{probabilistic1}) belongs to $\mathcal{T}^{(n)}$. Therefore, combining~(\ref{im8}) and~(\ref{im9}), we have 
\begin{align*}
\overline{G}(\lambda,\eps)&=\limsup_{n\to\infty}{G}_n(\lambda,\eps)\\
&\le \!\lambda \sqrt{V(P_0\|P_1)}\Phi^{-1}(\eps)  +(1\!-\!\lambda)\sqrt{V(P_1\|P_0)}\Phi^{-1}(\eps),
\end{align*}
as desired.

\section{Proof of Theorem~\ref{expectation-thm}}\label{sec:prf_exp}
In this section, we first present the tools used to prove Theorem~\ref{expectation-thm} in Subsection~\ref{tools}. We then establish upper and lower bounds on the error probabilities in Subsections~\ref{upperexpectation} and~\ref{lowerexpectation}, respectively.

\subsection{Auxiliary Tools}\label{tools}
In the proof of Theorem~\ref{expectation-thm}, we use the following results on the asymptotics of the first passage time from~\cite{HeydeAsymptotics, nonlinearrenewaltheory,GutAsymptotics}. For ease of reference, we include the results as follows. Theorem~\ref{asymptotic-nonlattice} characterizes the asymptotics of the {\em first passage times} of a stochastic process. 


Let $\{\alpha_i\}_{i=1}^\infty$ and $\{\beta_i\}_{i=1}^\infty$ be two increasing sequences of positive real numbers such that $\alpha_i\to \infty$ and $\beta_i\to \infty$ as $i\to \infty$.   Let $(\delta_i,T_i)$ be an SPRT with parameters $(\alpha_i,\beta_i)$. Recall that $Y_i=\log\frac{p_0(X_i)}{p_1(X_i)}$, $S_k=\sum_{i=1}^{k}Y_i$, and $T_n=\inf\{k\ge 1:S_{k}\notin [-\alpha_n,\beta_n]\}$.   The following two theorems, taken from~\cite[Theorem 3.1, pp.~31]{nonlinearrenewaltheory}, characterize  the asymptotics of the expected sample size and  the two types of error probabilities.
 
\begin{thm}\label{asymptotic-nonlattice}
Assume that  $\max\{\E_{P_1} [Y_1^2],\E_{P_0} [Y_1^2]\}<\infty$ and $Y_1$ is non-arithmetic. Then as $n\to\infty$,
\begin{align}
 \E_{P_0}[T_n]=\frac{\beta_n}{D(P_0\|P_1)}+\frac{A(P_0,P_1)}{D(P_0\|P_1)}+o(1),\notag
\end{align}
and
\begin{align}
 \E_{P_1}[T_n]=\frac{\alpha_n}{D(P_1\|P_0)}+\frac{\tilde{A}(P_0,P_1)}{D(P_1\|P_0)}+o(1).\notag
\end{align}
 \end{thm}

\begin{thm}\label{asymptoticserrornonlattice}
Assume that  $\max\{\E_{P_1} [Y_1^2],\E_{P_0} [Y_1^2]\}<\infty$ and $Y_1$ is non-arithmetic.
Then,
$$
\lim_{i\to\infty}P_{1|0}(\delta_i,T_i){e^{\alpha_i}}=e^{\tilde{B}(P_0,P_1)}$$
and
$$
\lim_{i\to\infty}P_{0|1}(\delta_i,T_i){e^{\beta_i}}=e^{B(P_0,P_1)}.\notag
$$
\end{thm}

The following lemma~\cite[Problem 3, pp.~369]{ferguson} characterizes the optimality of the SPRT. This lemma is also a simple consequence of Wald and Wolfowitz's work~\cite{WaldWolf}. 
\begin{lemma}\label{optimality}
Let $(\delta,T)$ be an SPRT. Let $(\tilde\delta,\tilde T)$ be any SHT such that
$$
\E_{P_0}[\tilde T]\le \E_{P_0}[T]\quad\mbox{and}\quad \E_{P_1}[\tilde T]\le \E_{P_1}[T].
$$
Then
$$
P_{0|1}(\delta,T)\le P_{0|1}(\tilde\delta,\tilde T)\quad\mbox{and}\quad P_{1|0}(\delta,T)\le P_{1|0}(\tilde\delta,\tilde T).
$$
\end{lemma}

\subsection{Upper Bound on the Error Probabilities}\label{upperbound}\label{upperexpectation}
For any $\eta>0$, let 
$$
\alpha_n=n D(P_1\|P_0)\bigg(1-\frac{\tilde{A}(P_0,P_1)}{nD(P_1\|P_0)}-\frac{\eta}{n}\bigg)
$$
and
$$
\beta_n=n D(P_0\|P_1)\bigg(1-\frac{A(P_0,P_1)}{nD(P_0\|P_1)}-\frac{\eta}{n}\bigg).
$$
Consider the SPRT $(\delta_n,T_n)$ with parameters $(\alpha_n,\beta_n)$. 
As $\alpha_n\to\infty$ and $\beta_n\to\infty$,   from Theorem~\ref{asymptoticserrornonlattice}, it follows that 
\begin{align}
P_{1|0}(\delta_n,T_n)&=P_{0}(S_T\le \alpha_n)= (e^{\tilde{B}(P_0,P_1)}+o(1))e^{-\alpha_n}.\notag
\end{align}
and similarly,
$$
P_{0|1}(\delta_n,T_n)= (e^{B(P_0,P_1)}+o(1))e^{-\beta_n}.
$$
We now show that $\E_{P_0} [T_n]\le n$ and $\E_{P_1} [T_n]\le n$. From Theorem~\ref{asymptotic-nonlattice}, it   follows that
\begin{align}\label{BST}
\E_{P_0}[T_n]&=\frac{\beta_n}{D(P_0\|P_1)}+\frac{A(P_0,P_1)}{D(P_0\|P_1)}+o(1)\notag\\
&=n-\eta+o(1),
\end{align}
where~(\ref{BST}) follows from the definition of $\beta_n$.
Similarly, $\E_{P_1}[T_n]\le n-\eta+o(1).$ Thus for sufficiently large $n$, we have that
$
\E_{P_0}[T_n]\le n
$
and
$
\E_{P_1}[T_n]\le n.
$
It then follows that
\begin{align}
\overline{F}(\lambda)&=\limsup_{n\to \infty}F_{n}(\lambda)\notag\\
&\le \lambda (\tilde{A}(P_0,P_1)+ \tilde{B}(P_0,P_1))\notag\\
&\hspace{0.5cm}+(1-\lambda)(A(P_0,P_1)+ B(P_0,P_1))\notag\\
&\hspace{0.5cm}+(\lambda D(P_0\|P_1)+(1-\lambda)D(P_1\|P_0))\eta.\notag
\end{align}
As $\eta>0$ is arbitrary, we conclude that
\begin{align}
\overline{F}(\lambda)&\le \lambda (\tilde{A}(P_0,P_1)+\tilde{B}(P_0,P_1))\notag\\
&+(1-\lambda)(A(P_0,P_1)+ B(P_0,P_1)).\notag
\end{align}

\subsection{Lower Bound on the Error Probabilities}\label{lowerexpectation}

For any $\eta>0$, let 
$$
\hat{\alpha}_n=n D(P_1\|P_0)\bigg(1-\frac{\tilde{A}(P_0,P_1)}{nD(P_1\|P_0)}+\frac{\eta}{n}\bigg)$$
and
$$\hat{\beta}_n=n D(P_0\|P_1)\bigg(1-\frac{A(P_0,P_1)}{nD(P_0\|P_1)}+\frac{\eta}{n}\bigg).
$$
Consider a sequence of SPRTs $\{(\hat{\delta}_n,\hat{T}_n)\}_{n=1}^{\infty}$ with parameters $\{(\hat{\alpha}_n,\hat{\beta}_n)\}_{n=1}^{\infty}$. 
Now we show that for sufficiently large $n$, we have that for $i=0,1$,
\begin{align}
\E_{P_i}[\hat{T}_n]\ge n+\frac{\eta}{2}.\notag
\end{align}
From Theorem~\ref{asymptotic-nonlattice} and the choice of $\hat{\beta}_n$, it follows that for sufficiently large $n$,
\begin{align}
\E_{P_0}[\hat{T}_n]&=\frac{\hat{\beta}_n}{D(P_0\|P_1)}+\frac{A(P_0,P_1)}{D(P_0\|P_1)}+o(1)\notag\\
&=n+\eta+o(1)\notag\\
&\ge n+\frac{\eta}{2}> n.\notag
\end{align}
Similarly, for sufficiently large $n$, 
$$
\E_{P_1}[\hat{T}_n]\ge n+\frac{\eta}{2}> n.
$$
Then from Lemma~\ref{optimality}, we conclude that for any SHT $(\delta_n, T_n)$ with $\max\{\E_{P_0}[T],\E_{P_1}[T]\}\le n$, 
\begin{align}
&P_{0|1}(\delta_n, T_n)\ge {P}_{0|1}(\hat{\delta}_n,\hat{T}_n)\label{largertest}\\
&P_{1|0}(\delta_n, T_n)\ge {P}_{1|0}(\hat{\delta}_n,\hat{T}_n). \nonumber
\end{align}
From Theorem~\ref{asymptoticserrornonlattice}, we have that 
\begin{align}\label{im11}
\log{P}_{1|0}(\hat{\delta}_n,\hat{T}_n)=  \tilde{B}(P_0,P_1)-\hat{\alpha}_n+\log(1+o(1))
\end{align}
and
\begin{align}
\log{P}_{0|1}(\hat{\delta}_n,\hat{T}_n)= B(P_0,P_1)-\hat{\beta}_n+\log(1+o(1)). \nonumber
\end{align}
Combining~(\ref{largertest}) and~(\ref{im11}), we have that
\begin{align}\label{im1}
&\log P_{1|0}(\delta_n, T_n)+nD(P_1\|P_0)\notag\\
&\hspace{1.5cm}\ge \log {P}_{1|0}(\hat{\delta}_n,\hat{T}_n)+nD(P_1\|P_0)\notag\\
&\hspace{1.5cm}\ge  \tilde{B}(P_0,P_1)+\tilde{A}(P_0,P_1) \notag\\
&\hspace{2cm}-\eta D(P_0\|P_1)+\log(1+o(1)).
\end{align}
Similarly, we have that
\begin{align}\label{im2}
&\log P_{0|1}(\delta_n, T_n)+nD(P_0\|P_1)\notag\\
&\hspace{1.5cm} \ge B(P_0,P_1)+A(P_0,P_1) \notag\\
&\hspace{2cm}-\eta D(P_1\|P_0)+\log(1+o(1)).
\end{align}
As $\lim_{x\to0}\log(1+x)=0$,   combining~(\ref{im1}) and~(\ref{im2}), we have that
\begin{align}
\underline{F}(\lambda)&\ge \liminf_{n\to \infty}F_{n}(\lambda)\notag\\
&\ge \lambda(\tilde{B}(P_0,P_1)+\tilde{A}(P_0,P_1))\notag\\
&\hspace{0.5cm}+(1-\lambda)( B(P_0,P_1)+A(P_0,P_1))\notag\\
&\hspace{0.5cm}-(\lambda D(P_0\|P_1)+(1-\lambda)(D(P_1\|P_0))\eta.\notag
\end{align}
Finally letting $\eta\to0^{+}$, we have 
\begin{align}
\underline{F}(\lambda)&\ge \lambda(\tilde{B}(P_0,P_1)+\tilde{A}(P_0,P_1))\notag\\
&\hspace{0.5cm}+(1-\lambda)( B(P_0,P_1)+A(P_0,P_1)),\notag
\end{align}
as desired.


\section{Conclusion and Future Work}
In this paper, we have quantified the backoff of the finite-length error exponents from their  asymptotic limits of $ D(P_1\| P_0)$ and $D(P_0\| P_1)$ for sequential binary hypothesis testing~\cite{WaldWolf}. We considered both the expectation and probabilistic constraint and concluded that the former is less stringent in the sense that the rate of convergence to $(  D(P_1\| P_0),D(P_0\| P_1))$ is faster. 

In the future, one may consider the following three natural extensions of this work. First, instead of the probabilistic constraint $\max_{i=0,1}P_i( T_n> n)\le\eps $ in \eqref{probabilistic1}, one can consider replacing the non-vanishing constant $\eps$ by a subexponentially decaying sequence. That is, one can consider the so-called {\em moderate deviations regime} \cite[Theorem~3.7.1]{Dembo}.  Second, one can consider a {\em classification}  counterpart~\cite{Gutman} of this problem in which in addition to the test sequence $X = \{X_i\}_{i=1}^\infty$, one also has training sequences $Y_0 = \{Y_{0,i}\}_{i=1}^\infty$ and $Y_1= \{Y_{1,i}\}_{i=1}^\infty$ drawn respectively from $P_0$ and $P_1$, which are assumed to be unknown. Finally, one may consider {\em universal} versions of sequential  hypothesis testing and quantify the price one has to pay as a result of complete incognizance of the generating distributions $P_0$ and~$P_1$.

\subsubsection*{Acknowledgements}
The authors would like to sincerely thank Dr.~Anusha Lalitha and Prof.~Tara Javidi for many helpful discussions during the initial phase of this work.

\bibliographystyle{IEEEtran}
\bibliography{reference}
\begin{IEEEbiographynophoto}{Yonglong Li}
is a research fellow at the Department of Electrical and Computer Engineering, National University of Singapore. He received the bachelor degree in Mathematics from Zhengzhou University in 2011 and the Ph.D.\ degree in Mathematics from the University of Hong Kong in 2015. From 2017 to 2019,  he was a postdoctoral fellow at the Center for Memory and Recording Research (CMRR), University of California, San Diego.
\end{IEEEbiographynophoto}

\begin{IEEEbiographynophoto}{Vincent Y.F. Tan}
(S'07-M'11-SM'15) was born in Singapore in 1981. He is currently a Dean's Chair Associate Professor in the Department of Electrical and Computer Engineering and the Department of Mathematics at the National University of Singapore (NUS).
He received the B.A.\ and M.Eng.\ degrees in Electrical and Information Sciences from Cambridge University in 2005 and the Ph.D.\ degree in Electrical Engineering and Computer Science (EECS) from the Massachusetts Institute of Technology (MIT)  in 2011.  His research interests include information theory, machine learning, and statistical signal processing.

Dr.\ Tan received the MIT EECS Jin-Au Kong outstanding doctoral thesis prize in 2011, the NUS Young Investigator Award in 2014,  the Singapore National Research Foundation (NRF) Fellowship (Class of 2018) and the NUS Young Researcher Award in 2019. He was also an IEEE Information Theory Society Distinguished Lecturer for 2018/9. He has authored a research monograph titled ``Asymptotic Estimates in Information Theory with Non-Vanishing Error Probabilities'' in the Foundations and Trends  in Communications and Information Theory Series (NOW Publishers). He is currently serving as an Associate Editor of the IEEE Transactions on Signal Processing and an Associate Editor of Machine Learning for the IEEE Transactions on Information Theory.
\end{IEEEbiographynophoto}
\end{document}